\documentclass[12pt]{article}

\usepackage[T1]{fontenc}
\usepackage[sc]{mathpazo}
\usepackage{amsmath}
\usepackage{amssymb}
\usepackage{enumerate}
\usepackage{amsthm}% theorems and proofs
\usepackage{amsfonts,mathrsfs}
\usepackage{hyperref}
\hypersetup{pdfpagemode=UseNone}

\newcommand{\tinyspace}{\mspace{1mu}}

\newcommand{\abs}[1]{\left\lvert\tinyspace #1 \tinyspace\right\rvert}

\newcommand{\setft}[1]{\mathrm{#1}}
\newcommand{\lin}[1]{\setft{L}\left(#1\right)}
\newcommand{\density}[1]{\setft{D}\left(#1\right)}
\newcommand{\unitary}[1]{\setft{U}\left(#1\right)}
\newcommand{\trans}[1]{\setft{T}\left(#1\right)}

\def\I{\mathbb{1}}

\newenvironment{mylist}[1]{\begin{list}{}{
    \setlength{\leftmargin}{#1}
    \setlength{\rightmargin}{0mm}
    \setlength{\labelsep}{2mm}
    \setlength{\labelwidth}{8mm}
    \setlength{\itemsep}{0mm}}}
    {\end{list}}

\def\ot{\otimes}

\newcommand{\out}[2]{| #1\rangle\langle #2 |}

\newcommand{\defeq}{\stackrel{\smash{\textnormal{\tiny def}}}{=}}

\newcommand{\pa}[1]{(#1)}
\newcommand{\Pa}[1]{\left(#1\right)}

\newcommand{\set}[1]{\{#1\}}

\newcommand{\ket}[1]{|#1\rangle}

\DeclareMathOperator{\trace}{Tr}
\newcommand{\ptr}[2]{\trace_{#1}\pa{#2}}
\newcommand{\Ptr}[2]{\trace_{#1}\Pa{#2}}
\newcommand{\tr}[1]{\ptr{}{#1}}
\newcommand{\Tr}[1]{\Ptr{}{#1}}

\def\cH{\mathcal{H}}
\def\cK{\mathcal{K}}

\def\rS{\mathrm{S}}

\newtheorem{thrm}{Theorem}[section]

\newtheorem{prop}[thrm]{Proposition}

\theoremstyle{definition}

\newtheorem{remark}[thrm]{Remark}

\numberwithin{equation}{section}

\newcounter{questionnumber}

\begin{document}

%=============================================================================%
\title{\large\bf Remark on the coherent information saturating its upper bound}
%=============================================================================%

\author{Lin Zhang\footnote{E-mail: godyalin@163.com; linyz@hdu.edu.cn}\\[1mm]
  {\it\small Institute of Mathematics, Hangzhou Dianzi University, Hangzhou 310018, PR~China}}

\date{}
\maketitle \mbox{}\hrule\mbox\\
\begin{abstract}

Coherent information is a useful concept in quantum information
theory. It connects with other notions in data processing. In this short
remark, we discuss the coherent information saturating its upper
bound. A necessary and sufficient condition for this saturation is
derived.

\end{abstract}
\mbox{}\hrule\mbox\\

%============================================================%
\section{Coherent information inequality}
%============================================================%

The fundamental problem in quantum error correction is to determine
when the effect of a quantum channel (trace-preserving completely
positive map) $\Phi \in \trans{\cH_B}$ acting on half of a pure
entangled state can be perfectly reversed. Define the \emph{coherent
information}
\begin{eqnarray}
I_c(\rho,\Phi) \defeq \rS(\Phi(\rho)) -
\rS(\I_A\ot\Phi(\out{\mathbf{u}_\rho}{\mathbf{u}_\rho})),
\end{eqnarray}
where $\ket{\mathbf{u}_\rho} =
\sum_{j}\sqrt{\lambda_j}\ket{x_j}\ot\ket{\lambda_j} \in \cH_A \ot
\cH_B$ is any \emph{purification} of $\rho = \sum_j
\lambda_j\out{\lambda_j}{\lambda_j}$.

In general, we have
\begin{eqnarray}
I_c(\rho,\Phi)\leqslant \rS(\rho).
\end{eqnarray}
It was shown that there exists a quantum channel $\Psi$ (see
\cite{Hayden}) such that
\begin{eqnarray}
I_c(\rho,\Phi) = \rS(\rho) \Longleftrightarrow
(\I_A\ot\Psi\circ\Phi)(\out{\mathbf{u}_\rho}{\mathbf{u}_\rho}) =
\out{\mathbf{u}_\rho}{\mathbf{u}_\rho}.
\end{eqnarray}
By the \emph{Stinespring dilation theorem}, we may assume that
$$
\Phi(\rho) =
\ptr{C}{U(\rho\ot\out{\epsilon}{\epsilon})U^\dagger},\quad
U\in\unitary{\cH_B\ot\cH_C}, \ket{\epsilon}\in\cH_C,
$$
which indicates that
\begin{eqnarray}
\I_A\ot\Phi(\out{\mathbf{u}_\rho}{\mathbf{u}_\rho}) &=&
\ptr{C}{(\I_A\ot
U)(\out{\mathbf{u}_\rho}{\mathbf{u}_\rho}\ot\out{\epsilon}{\epsilon})(\I_A\ot
U)^\dagger} \nonumber\\
&= &\Ptr{C}{\out{\Omega}{\Omega}},
\end{eqnarray}
where
$\ket{\Omega} = (\I_A\ot U)(\ket{\mathbf{u}_\rho} \ot
\ket{\epsilon})$. Now
$$
\out{\Omega}{\Omega} =(\I_A\ot
U)(\out{\mathbf{u}_\rho}{\mathbf{u}_\rho} \ot
\out{\epsilon}{\epsilon})(\I_A\ot U)^\dagger
$$
is a tripartite state
in $\density{\cH_A \ot \cH_B \ot \cH_C}$, it follows that
\begin{eqnarray*}
\ptr{C}{\out{\Omega}{\Omega}} & = & \I_A\ot\Phi(\out{\mathbf{u}_\rho}{\mathbf{u}_\rho}) \equiv \Omega_{AB},\\
\ptr{A}{\out{\Omega}{\Omega}} & = & U(\rho\ot\out{\epsilon}{\epsilon})U^\dagger \equiv \Omega_{BC},\\
\ptr{AC}{\out{\Omega}{\Omega}} & = & \Phi(\rho) \equiv \Omega_B,
\end{eqnarray*}
where $\Omega_{ABC} \equiv \out{\Omega}{\Omega}$. From the above
expressions, it is obtained that
\begin{eqnarray*}
\rS(\Omega_{ABC}) & = & 0, \\
\rS(\Omega_B) & = & \rS(\Phi(\rho))\\
\rS(\Omega_{BC}) & = & \rS(\rho),\\
\rS(\Omega_{AB}) & = &
\rS((\I_A\ot\Phi)(\out{\mathbf{u}_\rho}{\mathbf{u}_\rho}))
\end{eqnarray*}
Apparently, $I_c(\rho,\Phi) = \rS(\rho) \Longleftrightarrow
\rS(\Phi(\rho)) =
\rS((\I_A\ot\Phi)(\out{\mathbf{u}_\rho}{\mathbf{u}_\rho})) +
\rS(\rho)$, that is,
\begin{eqnarray*}
I_c(\rho,\Phi) = \rS(\rho) &\Longleftrightarrow&
\rS(\Omega_B) = \rS(\Omega_{AB}) + \rS(\Omega_{BC})\\
&\Longleftrightarrow& \rS(\Omega_B) - \rS(\Omega_C)
=\rS(\Omega_{BC}).
\end{eqnarray*}
It follows from Proposition~\ref{th:Araki-Lieb} in Appendix that this equation holds
if and only if
\begin{enumerate}[(i)]
\item $\cH_B$ can be factorized into the form $\cH_B = \cH_L \ot \cH_R$,
\item $\Omega_{BC} = \rho_L \ot \out{\psi}{\psi}_{RC}$ for $\ket{\psi}_{RC} \in \cH_R \ot \cH_C$.
\end{enumerate}
Hence $$ U(\rho\ot\out{\epsilon}{\epsilon})U^\dagger = \rho_L \ot
\out{\psi}{\psi}_{RC}\Longrightarrow \rho\ot\out{\epsilon}{\epsilon}
= U^\dagger \Pa{\rho_L \ot \out{\psi}{\psi}_{RC}}U.
$$
Clearly,
$\Omega_{ABC} = \out{\phi}{\phi}_{AL}\ot \out{\psi}{\psi}_{RC}$.
Thus
\begin{eqnarray}
\nonumber\out{\mathbf{u}_\rho}{\mathbf{u}_\rho} & = &
(\I_A\ot\Psi\circ\Phi)(\out{\mathbf{u}_\rho}{\mathbf{u}_\rho}) = (\I_A\ot\Psi)(\Omega_{AB})\\
\nonumber & = & (\I_A\ot\Psi)(\out{\phi}{\phi}_{AL}\ot \rho_R).
\end{eqnarray}
Since $\out{\Omega}{\Omega} = (\I_A\ot
U)(\out{\mathbf{u}_\rho}{\mathbf{u}_\rho} \ot
\out{\epsilon}{\epsilon})(\I_A\ot U)^\dagger$, it follows that
\begin{eqnarray}\label{eq:2}
\nonumber\out{\mathbf{u}_\rho}{\mathbf{u}_\rho} & = & \Ptr{C}{(\I_A\ot U)^\dagger\out{\Omega}{\Omega}(\I_A\ot U)}\\
 & = & \Ptr{C}{(\I_A\ot U)^\dagger \Pa{\out{\phi}{\phi}_{AL}\ot \out{\psi}{\psi}_{RC}} (\I_A\ot U)}.
\end{eqnarray}
The above equation gives that
$$
(\I_A\ot\Psi)\Pa{\out{\phi}{\phi}_{AL}\ot \rho_R} = \Ptr{C}{(\I_A\ot
U)^\dagger \Pa{\out{\phi}{\phi}_{AL}\ot \out{\psi}{\psi}_{RC}}
(\I_A\ot U)}.
$$
Given the state $\Omega_{AB} =
\I_A\ot\Phi(\out{\mathbf{u}_\rho}{\mathbf{u}_\rho})$, the recovery
procedure $\Psi$ is:
\begin{enumerate}[(i)]
\item preparing the state $\ket{\psi}_{RC}$ on $\cH_R\ot\cH_C$; thus
we have a state $\out{\phi}{\phi}_{AL}\ot \out{\psi}{\psi}_{RC}$.
\item next performing $U^\dagger$; we get
$$
(\I_A\ot U)^\dagger \Pa{\out{\phi}{\phi}_{AL}\ot
\out{\psi}{\psi}_{RC}} (\I_A\ot U).
$$
\item finally discarding the fixed ancillary state $\out{\epsilon}{\epsilon}$;
$$
\Ptr{C}{(\I_{A}\ot U)^\dagger \Pa{\out{\phi}{\phi}_{AL}\ot
\out{\psi}{\psi}_{RC}} (\I_{A}\ot U)}.
$$
\end{enumerate}
Note that $\I_A\ot\Phi(\out{u_\rho}{u_\rho}) =
\out{\phi}{\phi}_{AL}\ot \rho_R$ implies that
$$
\Phi(\rho) =
\rho_L\ot\rho_R.
$$

This indicates that the coherent information reaches its maximal
value if and only if the output state of the quantum channel $\Phi$
is a product state. Therefore we have the following theorem:
\begin{thrm}
Let $\rho\in\density{\cH}$ and $\Phi\in\trans{\cH}$ be a quantum
channel. The coherent
information achieves its maximum, that is, $I_c(\rho,\Phi) =
\rS(\rho)$ if and only if the following statements holds:
\begin{enumerate}[(i)]
\item the underlying Hilbert space can be decomposed as: $\cH = \cH_L\ot\cH_R$;
\item the output state of the quantum channel $\Phi$ is of a product form: $\Phi(\rho) = \rho_L\ot\rho_R$ for
$\rho_L\in\density{\cH_L},\rho_R\in\density{\cH_R}$.
\end{enumerate}
\end{thrm}

\begin{remark}
Consider a Kraus representation of a quantum channel $\Phi\in\trans{\cH}$ in its canonical Kraus
form: $\Phi = \sum_k \mathrm{Ad}_{M_k}$. For any $\rho \in \density{\cH}$, define
$$
\widehat\Phi(\rho) \defeq \sum_{i,j}\Tr{M_i\rho
M^\dagger_j}\out{i}{j}.
$$
If $\rho$ is purified as $\ket{\mathbf{u}_\rho} \in
\cH\ot\cK$ with $\dim(\cK)\geqslant\dim(\cH)$, then
$$
\rS(\widehat\Phi(\rho)) = \rS\Pa{(\Phi\ot\I_{\lin{\cK}})(\out{\mathbf{u}_\rho}{\mathbf{u}_\rho})}.
$$
Indeed, let $\rho
= \sum_k \lambda_k\out{\lambda_k}{\lambda_k}$ be its spectral decomposition,
\begin{eqnarray*}
\ket{\mathbf{u}_\rho} & \defeq & \sum_k\sqrt{\lambda_k}\ket{\lambda_k}\ot\ket{\lambda_k}, \\ \out{\mathbf{u}_\rho}{\mathbf{u}_\rho} &=&  \sum_{m,n}\sqrt{\lambda_m\lambda_n}\out{\lambda_m}{\lambda_n}\ot\out{\lambda_m}{\lambda_n}\\
\ket{\Omega} & \defeq &
\sum_{k,i}\sqrt{\lambda_k}M_i\ket{\lambda_k}\ot\ket{\lambda_k}\ot\ket{i}.
\end{eqnarray*}
Thus
$$
\out{\Omega}{\Omega} =
\sum_{m,n,i,j}\sqrt{\lambda_m\lambda_n}M_i\out{\lambda_m}{\lambda_n}M^\dagger_j\ot\out{\lambda_m}{\lambda_n}\ot\out{i}{j},
$$
which implies that
\begin{eqnarray*}
\Ptr{3}{\out{\Omega}{\Omega}} & = & \sum_{m,n,i}\sqrt{\lambda_m\lambda_n}M_i\out{\lambda_m}{\lambda_n}M^\dagger_i\ot\out{\lambda_m}{\lambda_n} \\
&=& \Phi\ot\I_{\lin{\cK}}(\out{\mathbf{u}_\rho}{\mathbf{u}_\rho}),\\
\Ptr{1,2}{\out{\Omega}{\Omega}} & = & \sum_{i,j} \tr{M_i\rho
M^\dagger_j}\out{i}{j} = \widehat\Phi(\rho).
\end{eqnarray*}
Clearly, $\rS\Pa{(\Phi\ot\I_{\lin{\cK}})(\out{\mathbf{u}_\rho}{\mathbf{u}_\rho})}$ is
independent of an arbitrary purification $\ket{\mathbf{u}_\rho}$ of $\rho$.
In fact, if $\ket{\mathbf{u}^{(1)}_\rho}$ and $\ket{\mathbf{u}^{(2)}_\rho}$ are any
two purification of $\rho$, then by Schimdt decomposition:
\begin{eqnarray*}
\ket{\mathbf{u}^{(1)}_\rho} & = & \sum_{k}\sqrt{\lambda_k}\ket{\lambda_k}\ot\ket{x_k},\\
\ket{\mathbf{u}^{(2)}_\rho} & = &
\sum_{k}\sqrt{\lambda_k}\ket{\lambda_k}\ot\ket{y_k},
\end{eqnarray*}
it is seen that there exists an isometry operator $U$ such that
$U\ket{x_k} = \ket{y_k}$ for each $k$, moreover $\ket{\mathbf{u}^{(2)}_\rho}
= (\I\ot U)\ket{\mathbf{u}^{(1)}_\rho}$. Now
$\out{\mathbf{u}^{(2)}_\rho}{\mathbf{u}^{(2)}_\rho} = (\I\ot
U)\out{\mathbf{u}^{(1)}_\rho}{\mathbf{u}^{(1)}_\rho}(\I\ot U)^\dagger$, which implies
that
\begin{eqnarray*}
(\Phi\ot\I)(\out{\mathbf{u}^{(2)}_\rho}{\mathbf{u}^{(2)}_\rho}) &=& (\I\ot U) (\Phi\ot\I)(\out{\mathbf{u}^{(1)}_\rho}{\mathbf{u}^{(1)}_\rho})(\I\ot U)^\dagger,\\
\rS\Pa{(\Phi\ot\I)(\out{\mathbf{u}^{(1)}_\rho}{\mathbf{u}^{(1)}_\rho})} &=&
\rS\Pa{(\Phi\ot\I)(\out{\mathbf{u}^{(2)}_\rho}{\mathbf{u}^{(2)}_\rho})}.
\end{eqnarray*}
\end{remark}

%=========================================================================%
\section{Appendix}
%=========================================================================%

\subsection{The saturation of the strong subadditivity inequality}

\begin{prop}[\cite{Hayden}]\label{th:SSAwithequality}
A state $\rho_{ABC}\in\density{\cH_A\ot\cH_B\ot\cH_C}$ saturating
the strong subadditivity inequality, i.e.,
$$
\rS(\rho_{AB}) + \rS(\rho_{BC}) = \rS(\rho_{ABC}) + \rS(\rho_B)
$$
if and only if there is a decomposition of system $B$ as
$$
\cH_B = \bigoplus_j \cH_{b^L_j}\ot\cH_{b^R_j}
$$
into a direct (orthogonal) sum of tensor products, such that
$$
\rho_{ABC} = \bigoplus_j \lambda_j\rho_{Ab^L_j}\ot \rho_{b^R_jC},
$$
where $\rho_{Ab^L_j}\in\density{\cH_A\ot\cH_{b^L_j}}$ and
$\rho_{b^R_jC}\in\density{\cH_{b^R_j}\ot\cH_C}$, and
$\set{\lambda_j}$ is a probability distribution.
\end{prop}

\subsection{The saturation of Araki-Lieb inequality}

The following proposition can be seen as a characterization of the
saturation of Araki-Lieb inequality:
\begin{eqnarray}
\abs{\rS(\rho_B)-\rS(\rho_C)}\leqslant \rS(\rho_{BC}).
\end{eqnarray}
For the readers' convenience, we copy the proof here.
\begin{prop}[\cite{Zhang}]\label{th:Araki-Lieb}
Let $\rho_{BC}\in\density{\cH_B\ot\cH_C}$. The reduced states are
$\rho_B =\ptr{C}{\rho_{BC}}, \rho_C =\ptr{B}{\rho_{BC}}$,
respectively. Then $\rS(\rho_{BC}) = \rS(\rho_B) - \rS(\rho_C)$ if
and only if
\begin{enumerate}[(1)]
\item\label{1} $\cH_B$ can be factorized into the form $\cH_B = \cH_L \ot \cH_R$,
\item\label{2} $\rho_{BC} = \rho_L \ot \out{\psi}{\psi}_{RC}$ for $\ket{\psi}_{RC} \in \cH_R \ot \cH_C$.
\end{enumerate}
\end{prop}

\begin{proof}
The sufficiency of the condition is immediate. The proof of
necessity is presented as follows: Assume that $\rS(\rho_{BC}) =
\rS(\rho_B) - \rS(\rho_C)$. The bipartite state $\rho_{BC}$ can be
purified into a tripartite state $\ket{\Omega_{ABC}} \in \cH_A \ot
\cH_B \ot \cH_C$, where $\cH_A$ is a reference system. Denote
$\rho_{ABC} = \out{\Omega_{ABC}}{\Omega_{ABC}}$. We have
\begin{eqnarray*}
\Ptr{AB}{\rho_{ABC}} = \rho_C,\quad \Ptr{AC}{\rho_{ABC}} = \rho_B,\\
\Ptr{C}{\rho_{ABC}} = \rho_{AB},\quad \Ptr{A}{\rho_{ABC}} =
\rho_{BC}.
\end{eqnarray*}
Now since $\rS(\rho_{ABC}) = 0$, it follows that $\rS(\rho_C) =
\rS(\rho_{AB})$. Thus we have
$$
\rS(\rho_{AB}) + \rS(\rho_{BC}) = \rS(\rho_B) = \rS(\rho_B) +
\rS(\rho_{ABC}),
$$
which, by Proposition~\ref{th:SSAwithequality}, implies that
\begin{enumerate}[(i)]
\item $\cH_B$ can be factorized into the form $\cH_B = \bigoplus_{k=1}^K \cH_{b^L_k} \ot \cH_{b^R_k}$,
\item $\rho_{ABC} = \bigoplus_{k=1}^K \lambda_k \rho_{Ab^L_k} \ot \rho_{b^R_kC}$ for
$\rho_{Ab^L_k} \in \density{\cH_A \ot \cH_{b^L_k}}$ and
$\rho_{b^R_kC} \in \density{\cH_{b^R_k} \ot \cH_C}$, where
$\set{\lambda_k}$ is a probability distribution.
\end{enumerate}
Clearly,
$$
\rS(\rho_{BC}) = \rS(\rho_B) - \rS(\rho_C)\Longrightarrow\rS(\rho_A)
+ \rS(\rho_C) = \rS(\rho_{AC}).
$$
But
$$
\rS(\rho_A) + \rS(\rho_C) =
\rS(\rho_{AC})\Longleftrightarrow\rho_{AC} = \rho_A \ot \rho_C.
$$
From the expression
$$
\rho_{ABC} = \bigoplus_{k=1}^K \lambda_k \rho_{Ab^L_k} \ot
\rho_{b^R_kC},
$$
it follows that
$$
\rho_{AC} = \sum_{k=1}^K \lambda_k \rho_{A,k} \ot \rho_{C,k}.
$$
Combining all the facts above mentioned, we have
$$
K=1,
$$
i.e., the statement \eqref{1} in the present theorem holds. Hence
$$
\rho_{ABC} = \rho_{AL} \ot \rho_{RC}
$$
for $\rho_{AL} \in \density{\cH_A \ot \cH_L}$ and $\rho_{RC} \in
\density{\cH_R \ot \cH_C}$, which implies that both $\rho_{AL}$ and
$\rho_{RC}$ are pure states since $\rho_{ABC}$ is pure state.
Therefore
\begin{eqnarray}
\rho_{BC} = \Ptr{A}{\rho_{AL}} \ot \rho_{RC} = \rho_L \ot
\out{\psi}{\psi}_{RC}
\end{eqnarray}
for $\ket{\psi}_{RC} \in \cH_R \ot \cH_C$, i.e., the statement
\eqref{2} holds. This completes the proof.
\end{proof}

\begin{remark}
The result in Proposition~\ref{th:Araki-Lieb} is employed to study the saturation of the upper bound of quantum discord in \cite{Xi}. Later on, E.A Carlen gives an elementary proof about this result in \cite{Carlen}.
\end{remark}

\end{document}